\documentclass[letterpaper, 10 pt, conference]{ieeeconf}  

\IEEEoverridecommandlockouts                              
\usepackage{cite}
\usepackage{amsmath,amssymb,amsfonts}
\usepackage{algorithmic}
\usepackage{graphicx}
\usepackage{textcomp}
\usepackage{xcolor}
\usepackage{wrapfig}
\usepackage{times}
\usepackage{tikz}
\usepackage{mathtools}
\usepackage{url}
\usetikzlibrary{automata, positioning}
\usepackage{hyperref}
\usepackage{xurl}
\newtheorem{definition}{Definition}
\newtheorem{theorem}{Theorem}

\newtheorem{example}{Example}
\newtheorem{lemma}{Lemma}

\definecolor{darkgreen}{RGB}{0,127,0}

\title{\LARGE \bf
Control Synthesis in Partially Observable Environments for Complex Perception-Related Objectives 
}

\author{Zetong Xuan and Yu Wang
\thanks{Zetong Xuan and Yu Wang are with the Department of Mechanical \& Aerospace Engineering,
        University of Florida, Gainesville, FL 32611, USA; 
         {\tt\small  z.xuan@ufl.edu, yuwang1@ufl.edu}}%
}

\usepackage{algorithm}
\usepackage{algorithmic}
\usepackage{multicol}
\usepackage[normalem]{ulem}

\usepackage{todonotes}
\begin{document}

\maketitle
\thispagestyle{empty}
\pagestyle{empty}

\begin{abstract}
Perception-related tasks often arise in autonomous systems operating under partial observability. 
This work studies the problem of synthesizing optimal policies for complex perception-related objectives in environments modeled by partially observable Markov decision processes. 
To formally specify such objectives, we introduce \emph{co-safe linear inequality temporal logic} (sc-iLTL), which can define complex tasks that are formed by the logical concatenation of atomic propositions as linear inequalities on the belief space of the POMDPs. 
Our solution to the control synthesis problem is to transform the \mbox{sc-iLTL} objectives into reachability objectives by constructing the product of the belief MDP and a deterministic finite automaton built from the sc-iLTL objective. 
To overcome the scalability challenge due to the product, we introduce a Monte Carlo Tree Search (MCTS) method that converges in probability to the optimal policy. 
Finally, a drone-probing case study demonstrates the applicability of our method. 
\end{abstract}


\section{Introduction}
\label{sec:intro}

Perception-related tasks are commonly encountered in the operation of autonomous systems in partially observable environments. 
For example, consider a drone with a limited field of view and imperfect onboard sensors. 
Tasks such as confidently detecting a target or avoiding obstacles are perception-related because they rely on estimation of the environment rather than on direct observation. 
Beyond that, many real-world applications require complex rule\mbox{-}based tasks that consist of sequentially ordered perception-related subtasks. 
For instance, a surveillance drone might be required to continuously monitor a moving target (liveness) while simultaneously avoiding obstacles (safety). 
These sequential subtasks necessitate a precise formulation of each {perception\mbox{-}related objective} as well as their correct ordering to achieve the overall mission. 

In practice \cite{thrun2005probabilistic,kurniawati2022partially}, the perception in such partially observable environments is typically captured by \emph{beliefs} \cite{kaelbling1998planninga} in the partially observable Markov decision process (POMDP). 
A POMDP generalizes a Markov decision process (MDP) by assuming that the true MDP states, referred to as hidden states, are partially observable through a probabilistic relation to a set of observations. 
The perception of the environment is represented by the beliefs, which are the optimal estimation of the probability distribution over hidden states based on all previous and current observations. 
Thus, perception-related tasks can be expressed by reaching or avoiding certain regions over the space of all beliefs (belief space). 

In POMDPs, defining perception-related rule-based tasks requires formal logic over beliefs.
In this work, we introduce \emph{co-safe linear inequality temporal logic} (sc-iLTL) \cite{kwon2004linear} to describe perception-related objectives. 
When the rule-based task involves the sequential ordering of perception-related events, sc-iLTL specifies the corresponding temporal relations over beliefs. 
Sc-iLTL modifies standard linear temporal logic (LTL)~\cite{baier2008principles} 
i) by expressing each atomic proposition as a linear inequality over the belief space, making the formulas better suited for perception-related tasks; and  
ii) by restricting to the co-safe fragment~\cite{kupferman2001model} of iLTL~\cite{kwon2004linear}, which aligns naturally with finite-horizon tasks. 
This second modification adopts co-safe LTL~\cite{kupferman2001model}.  
Here, sc-LTL refers to the co-safe fragment of LTL, which consists of formulas satisfiable by finite path prefixes, making it well-suited for finite-horizon tasks. 
In contrast, general LTL formulas may require infinite-length paths to determine satisfaction, which is less appropriate in such contexts.

Synthesizing an optimal policy for an sc-iLTL objective is challenging due to scalability issues. 
This problem is as hard as solving a reachability problem in the continuous belief space, which is known to be hard due to the large state space~\cite{lavaei2022automated}. 
The reachability problem is formulated by augmenting the original model with an automaton that tracks the satisfaction status of the sc-iLTL objective.
However, as we define the objectives directly over the space of all beliefs, the resulting reachability problem cannot be solved efficiently using point-based methods~\cite{bouton2020pointbased}, which rely on the $\alpha$-vector representation of the value function.
Such representation assumes the value function is piecewise linear or convex in the belief space. This key assumption does not hold in our case.

To address the scalability challenge, we propose using Monte Carlo tree search (MCTS) \cite{silver2010montecarlo}. 
Methods like MCTS or real time dynamic programming (RTDP)\cite{geffner1998solving} mitigate the scalability issue by leveraging the structured dynamics of belief updates. 
Specifically, given finite hidden states, the number of possible next beliefs remains finite even though the belief space is continuous.
MCTS applies best-first search to a finite number of beliefs that are only reachable from the initial beliefs, thus achieving scalability similar to efficient point-based methods. 
We prefer MCTS over RTDP because MCTS avoids explicit discretization, making its complexity independent of the size of the belief space. 

Our main contribution is to propose an MCTS method for synthesizing an optimal policy.  
By constructing the belief MDP, which captures the transitions of beliefs, we first lift the control synthesis problem to the belief space. 
Then, we build a product belief MDP by augmenting the belief MDP with a deterministic finite automaton (DFA) encoding the satisfaction status of the sc-iLTL objective. 
This product construction transforms the history-dependent sc-iLTL objective into a reachability objective on the product belief MDP.
Consequently, we show a Markovian policy that maximizes the reachability probability in the product belief MDP corresponds to a non-Markovian policy that maximizes the sc-iLTL satisfaction probability in the original POMDP.
Finally, we apply the MCTS method to compute the optimal Markovian policy, which can then be mapped back to an optimal non-Markovian policy for the original POMDP.

\paragraph*{Related work}
Our work relates to control synthesis for LTL objectives in POMDPs. 
Most existing works~\cite{chatterjee2015qualitative, ahmadi2020stochastic, sharan2014formal,bouton2020pointbased} define LTL objectives over hidden states, thus \mbox{unsuitable} for perception-related objectives. 
Some studies~\cite{vasile2016control, haesaert2018temporal} have explored perception-related objectives defined as distribution temporal logic (DTL)~\cite{jones2013distribution}, where nonlinear inequalities map beliefs into atomic propositions. 
However, these works focus on linear Gaussian POMDPs, a subclass of POMDPs where beliefs are represented as a Gaussian distribution characterized by its mean and covariance. 
In contrast, we target general POMDPs where beliefs cannot be parametrized this way.

\section{Preliminaries} \label{sec:prelim}

This section introduces POMDP, belief MDP. 
POMDP models the probabilistic transition of the hidden states and the probabilistic outcome of receiving an observation, given an action. 
Belief MDPs with continuous belief space capture the transition of optimal estimations of hidden states. 
 
\subsection{POMDP}

We model an autonomous system's dynamics in the unknown and partially observable environment by a POMDP.
\begin{definition}
    A POMDP is a tuple $\mathcal{M}_\mathcal{P} = (S, A, T_\mathcal{P}, s_\mathrm{init}, \Omega, O)$, where 
    i) $S$ is a finite set of hidden states, 
    ii) $A$ is a finite set of actions, 
    iii) $T_\mathcal{P}: S \times A \times S \to [0,1]$ is the transition probability function such that for all $s\in S$, we have 
    $\sum_{s'\in S}T_\mathcal{P}(s,a,s') = 1$, if $a \in A(s)$ and $\sum_{s'\in S}T_\mathcal{P}(s,a,s') = 0$, if $a \notin A(s)$, 
    iv) $s_\mathrm{init}$ is the initial distribution of the hidden state, 
    v) $\Omega$ is a finite set of observations, 
    where $\Omega(s)$ denotes the set of possible observations in the state $s\in S$,
    vi) $O:S\times \Omega \to [0,1]$ is the observation probability function such that for all $s\in S$, we have
    $\sum_{o\in \Omega}O(s,o) = 1$. 
\end{definition}

We call a sequence of states $\sigma_\mathcal{P}: \mathbb{N} \to S$ a \textit{path} of the POMDP if for any $t \in \mathbb{N}$, there exists $a \in A$ such that $T_\mathcal{P}({\sigma_\mathcal{P}}_{t}, a, {\sigma_\mathcal{P}}_{t+1}) > 0$. 
 
\subsection{Belief MDP} \label{sub:belief mdp}

The perception of the hidden POMDP states can be derived from the past actions and observations as follows.
Let $h_t = \{a_0, o_0,... a_{t-1}, o_{t-1}\}$ with $h_0 = \emptyset$ be a \textit{history} of actions and observations. The \emph{belief} (i.e., probabilistic estimation) of the current hidden states can be derived by the map $b_t = \mathcal{B}(h_t)$, formally defined as
\begin{align} \label{eq:belief update}
b_t (s) = \frac{ O(s, o_t)\sum_{s' \in S}b_{t-1}(s') T_\mathcal{P}(s',a_t,s)} {\sum_{s\in S}O(s, o_t)\sum_{s' \in S}b_{t-1}(s') T_\mathcal{P}(s',a_t,s)},
\end{align}
for any $s\in S$ and initial estimation $b_0(s) = s_\mathrm{init}(s)$.

Given the history, the transition of belief states with respect to actions can be treated as a continuous state MDP called \emph{belief MDP}. 

\begin{definition} \label{def:belief mdp}
The belief MDP $\mathcal{M}_\mathcal{B}$ of the POMDP $(S, A, \allowbreak T_\mathcal{P},  s_\mathrm{init}, \Omega, O)$ is defined by $\mathcal{M}_\mathcal{B} = (B, A, T_\mathcal{B}, b_0)$ where
i) $B := \Delta(S)$ is the belief space, 
ii) $A$ is the set of actions, 
iii) $T_\mathcal{B}:B\times A \times B \to [0,1]$ is the transition probability function 
$T_\mathcal{B} (b, a, b') = \sum_{o \in O} \sum_{s' \in S} \sum_{s \in S} \eta(b, o, b') O (s', o) T_\mathcal{P}(s, a, s') b(s)$, 
\begin{align}
\eta(b, o, b') = \bigg\{
\begin{array}{ll}
1, & \text {if belief update~\eqref{eq:belief update}  for } b, o \text{\, returns } b' \notag 
\\ 0, & \text {otherwise,}\end{array}          
\end{align}
iv) $b_0 = s_\mathrm{init}$.  
\end{definition}
A sequence of belief states $\sigma: \mathbb{N} \to \Delta(S)$ is a \textit{belief path } of the belief MDP if for any $t \in \mathbb{N}$, there exists $a \in A$ such that $T_\mathcal{B}(\sigma_{t}, a, \sigma_{t+1})) > 0$. 

Typically, a policy on POMDP is synthesized in belief space \cite{kurniawati2022partially}. 
Note that a \textit{Markovian policy} $\pi(b_t, t) \in A$ only uses the current belief state and time step as input.
In contrast, a \textit{non-Markovian policy} may use the entire sequence of previous beliefs as input.

\section{Problem Formulation and Main Result}\label{sec:pf}
We now introduce sc-iLTL as control objectives. 
To motivate its use, we first present an example to show that some objectives can be specified on the belief path $\sigma$ but are hard to specify on the hidden state path $\sigma_\mathcal{P}$. 
We then formally define our control objective and state the main result of solving it. 
 
\subsection{Motivating example}\label{sec:example}
The common way we specify the control objectives on a POMDP is through hidden states, which may not be expressive enough for a perception-related task.
\begin{example}
Consider a drone-probing task in a grid world (inspired by \cite{svorenova2015temporal}). The drone needs to locate a ground target on the grid using an imperfect sensor. 
The \textbf{sensor} only provides the respective quadrant of the target within the field of view, that is, adjacent to or under the drone. 
The observation set is $\Omega = \{ \mathrm{SW},\mathrm{NW},\mathrm{NE},\mathrm{SE},\mathrm{None}\}$, where the first four stand for the respective quadrant and ``$\mathrm{None}$" means the target is not in the field of view. Suppose the target is at the adjacent grid north of the drone, then the sensor returns ``$\mathrm{NE}$" or ``$\mathrm{NW}$" with equal probability, and if the target is under the drone, then the sensor returns an observation in $\{ \mathrm{SW},\mathrm{NW},\mathrm{NE},\mathrm{SE}\}$ with equal probability. 

In the POMDP model, the hidden state encodes the positions of both the drone and the ground target.  
Whereas the belief state is the optimal estimation of the distribution of ground target position.
We can specify such a task as increasing the maximum element of the belief state to a confidence threshold as $\Vert b \Vert_\infty \ge c$.
\end{example}

This task is difficult to define over hidden states, although it is naturally defined over beliefs,
since $\Vert b \Vert_\infty \ge c \notag$ requires the drone to infer the unknown location of the target by belief update \eqref{eq:belief update}. 
It is hard to encode the belief update using traditional control objectives like reward $r(s,a):S\times A \to \mathbb{R}$ or LTL with atomic propositions defined on hidden states. 
 
\subsection{Perception-related objectives in sc-iLTL}
We express the perception-related objective using the \mbox{sc-iLTL} formula \cite{kwon2004linear} constructed via atomic propositions defined on belief states. 
The definition is similar to the standard LTL formula \cite{baier2008principles} except for atomic propositions. 
An iLTL formula is derived recursively from the rules
\begin{align}
    \label{def:iltl syntax}
    \varphi &\Coloneqq \mathrm{true}
    \mid \mathrm{ineq}
    \mid \neg \varphi
    \mid \varphi_1 \land \varphi_2
    \mid \bigcirc \varphi 
    \mid \varphi_1 \mathcal{U} \varphi_2, \, \mathrm{ineq}\in \Lambda \notag,
\end{align}
where (next) $\bigcirc \varphi$ holds if $\varphi$ holds in the next ``step,'' and (until) $\varphi_1 \mathcal{U} \varphi_2$ holds if $\varphi_1$ holds at all moments until a future moment for which $\varphi_2$ holds. 
Other propositional and temporal operators can be derived from previous operators, e.g., (or) $\varphi_1 \vee \varphi_2 \coloneqq \neg(\neg \varphi_1 \wedge \neg \varphi_2)$, (eventually) $\lozenge \varphi \coloneqq {\rm true} \cup \varphi$.

We focus on the co-safe fragment of iLTL formulas, which can be verified via finite prefixes.  
We write $\sigma \models \varphi$ to denote that the infinite belief path $\sigma$ satisfies $\varphi$.  
An iLTL formula $\varphi$ is said to be co-safe~\cite{lacerda2014optimal} if, for every infinite sequence $\sigma \models \varphi$, there exists a finite prefix $\sigma_{\leq k} = b_0 b_1 \cdots b_k$, $k \in \mathbb{N}$, such that $\sigma_{\leq k} \cdot \sigma' \models \varphi$ for any infinite sequence $\sigma'$.
Sc-iLTL restricts formulas so that $\neg$ is applied only to atomic propositions, and only the temporal operators $\bigcirc$ (next), $\mathcal{U}$ (until), and $\lozenge$ (eventually) are allowed.
Let $\sigma: \mathbb{N} \to \Delta(S)$ be a belief path. The satisfaction of an sc-iLTL formula follows the standard semantics of co-safe LTL~\cite{lacerda2014optimal}, except for the interpretation of atomic propositions, 
\begin{align}
    \sigma \models \mathrm{ineq} \quad \text{iff} \quad p^T \sigma_0 > c, \quad \text{with } p \in \mathbb{R}^{|S|},\; c \in \mathbb{R}. \notag
\end{align}

\noindent\textbf{Problem Formulation}: 
Given a POMDP $\mathcal{M}_\mathcal{P} = (S, A, \allowbreak T_\mathcal{P}, s_\mathrm{init}, \Omega, O)$ and an sc-iLTL objective $\varphi$ over the belief path $\sigma$. Find an optimal policy $\pi$ (which may depend on all history observations and actions) that maximizes the satisfaction probability $Pr_\pi(\sigma \models \varphi)$. 

\subsection{Overview of technical solutions}
Solving this control synthesis problem, where the objective is defined over the belief space, is challenging due to the non-Markovian nature of the optimal policy and the scalability issues involved in finding it.
The optimal policy is typically non-Markovian, as the sc-iLTL formula specifies a sequential ordering of events.
At the same time, control synthesis in the continuous belief space inherently suffers from scalability issues.

We can transform the sc-LTL control problem into a maximum reachability problem, for which a Markovian policy can be optimal.
This transformation is achieved by constructing the product of the belief MDP and a graph representation of the sc-iLTL objective.
Specifically, we use DFA as a graph representation of an sc-iLTL objective, thus simplifying the complex rule-based objective into reachability. 
Given an sc-iLTL objective $\varphi$, we can construct a DFA $\mathcal{A}_\varphi$ (with labels $\Sigma=2^{\Lambda}$) such that a belief path $\sigma \models \varphi$ if and only if $\sigma$ is accepted by the DFA $\mathcal{A}_\varphi$.
\begin{definition}
A DFA is a tuple $\mathcal{A} = (Q, \Sigma,\delta, q_0, F)$  where
i) $Q$ is a finite set of automaton states, 
ii) $\Sigma$ is a finite set of alphabets, 
iii) $\delta: Q \times  \Sigma \to 2^Q$ is a (partial) transition function, 
iv) $q_\mathrm{init} \in Q$ is the initial automaton state, 
v) $F\subseteq Q$ is a set of final automaton states.
\end{definition}

The DFA verifies the belief path by a labeling function $L : B \to 2^\Lambda$, where $ \mathrm{ineq} \in L(b)$ if and only if $\mathrm{ineq}$ holds on $b$. 
The labeling on the belief path moves the automation states.  
A belief path $\sigma$ is accepted by $\mathcal{A}_\varphi$ if and only if there exists $k\in \mathbb{N}$ such that the prefix $\sigma_{\leq k}$ moves the initial automaton state $q_0$ to a final state $q\in F$. 

We use MCTS \cite{silver2010montecarlo} to address the scalability issue when solving the transformed maximum reachability problem. 
Instead of exploring the entire belief space, MCTS focuses only on beliefs that are reachable from the initial belief.
This is achieved by simulating a history MDP, which effectively reduces the continuous belief space to a discrete set of reachable beliefs.
\begin{definition}
The history MDP $\mathcal{M}_\mathcal{H}$ \cite{li2024reinforcement} of the POMDP $(S, A, T_\mathcal{P}, s_\mathrm{init}, \Omega, O)$ is defined as $\mathcal{M}_\mathcal{H} = (H, A, T_\mathcal{H}, h_0)$, where 
i) $H$ is the discrete state space consisting of all possible history states $h_t$,  
ii) $A$ is the action set,  
iii) $T_\mathcal{H} : H \times A \times H \to [0,1]$ is the transition probability function. For $h \in H$, $a \in A$, and $o \in \Omega$, the next history is defined as $h' = h \cdot a \cdot o$, where $\cdot$ denotes sequence concatenation, and $
T_\mathcal{H}(h, a, h') = \sum_{s' \in S} \sum_{s \in S} O(s', o)\, T_\mathcal{P}(s, a, s')\, \mathcal{B}(h)(s)$. 
iv) $h_0 = \emptyset$ is the initial (empty) history state. 
\end{definition}
Although the number of history states grows exponentially with the length of the belief path, MCTS mitigates this issue via best-first search. 
In particular, the sample complexity is independent of the size of the state space. 

Finally, the main result is formally stated below with proof deferred to the end of Section~\ref{sec:MCTS}. 

\begin{theorem} \label{thm:4}
    Algorithm \ref{alg:example} finds the optimal policy on the POMDP that maximizes the satisfaction probability of the given sc-iLTL objective. 
\end{theorem}

\section{MCTS for Sc-iLTL}
In this section, we explain our algorithm to solve the control synthesis problem for sc-iLTL objectives. 
First, we transform the sc-iLTL control problem into a reachability problem where a Markovian policy can be optimal by constructing the product belief MDP. 
Then, we apply MCTS to find the optimal Markovian policy.

\subsection{Product belief MDP}
We translate the control problem for the POMDP into a reachability problem on the product belief MDP. 
The transitions of the product MDP $\mathcal{M}_\mathcal{B}^\times $ are defined by combining the transitions of the belief MDP and the DFA. 
Each belief path $\sigma\models \varphi$ will have a corresponding path $\sigma^\times \models \lozenge F^\times$ on the product belief MDP $\mathcal{M}^{\times}$ and vice versa. 
\begin{definition} \label{def:product mdp}
A product belief MDP $\mathcal{M}_\mathcal{B}^{\times} = ( B^{\times}, \allowbreak A^{\times} , P^{\times},b_0^{\times}, F^{\times})$ of a belief MDP $\mathcal{M}_\mathcal{B} = (B, \allowbreak A, T_\mathcal{B}, b_0)$ and a DFA ${\mathcal{A}}= (Q,\Sigma,\delta,q_\mathrm{init},F)$ is defined by
i) the set of states $B^\times = B \times Q \cup \mathrm{sink}$, 
ii) the set of actions ${A}^\times = A$, 
iii) the set of final states $F^\times=\{\langle b,q \rangle \in B^{\times} |q\in F\}$, 
iv) the initial state $b_0^\times=\langle b_0,q_\mathrm{init} \rangle$, 
v) the transition probability function 
\begin{align}
    &T^\times_\mathcal{B}(b^\times ,a, b^{\times'}) \notag \\
    & = 
    \begin{cases}
    T_\mathcal{B}(b,a,b') &  b^\times = \langle b,q\rangle\notin F^\times,\, b^{\times'}=\langle b', q'\rangle, \\ & q' = \delta(q,L(b))\\
    1 &   b^\times \in F^\times,\, b^{\times'}=\mathrm{sink}\\
    1 &   b^\times = \mathrm{sink},\, b^{\times'}=\mathrm{sink}\\
    0 & \mathrm{else}.
\end{cases} 
\end{align} 
\end{definition}
Once $F^\times$ is met, the path on the product belief MDP falls to the $\mathrm{sink}$. 
This $\mathrm{sink}$ simplifies transitions after reaching $F^\times$, since we are only interested in the finite prefix of $\sigma^\times$.



\vspace{-5pt}
\subsection{Reachability problem on the product belief MDP}
Now we formulate the reachability problem on the product belief MDP. 
We introduce an undiscounted reward function 
$r(b^\times) = 1 \text{ if } b^\times \in F^\times, \text{ else } 0$.
such that the value function $V^{\pi^\times}(b^\times_0) := \mathbb{E}_\pi \sum_{i=0}^\infty r(b^\times_i)$, which is the expected return of the cumulative reward along all paths starting at $b^\times_0$, 
equals the reachability probability. 

\begin{lemma}\label{lem: V to R}
Given a policy $\pi^\times$, belief state $b$ and DFA state $q$, it holds that
\begin{align}
    V^{\pi^\times}(\langle b,q \rangle) = Pr_{\pi^\times}(\sigma^\times \models \lozenge F^\times | \sigma^\times_0 = \langle b,q \rangle). 
\end{align}
\end{lemma}
\begin{proof}
Along a belief path visiting $F^\times$, reward $1$ will be collected only once. 
Along a belief path not visiting $F^\times $, no positive reward will be collected. 
Thus  $V^{\pi^\times}(b^\times_0) := \mathbb{E}_\pi \sum_{i=0}^\infty r(b^\times_i)$
equals the reachability probability to $F^\times$. 
\end{proof}

Given a policy, the value function satisfies the recursive formulation as the Bellman equation, when $b^\times = \mathrm{sink}$, $V^{\pi^\times}(b^\times) = 1$, 
when $b^\times \ne \mathrm{sink}$, 
\begin{align}\label{eqn: value}
    V^{\pi^\times}(b^\times) = \textstyle\sum_{b^{\times'}} T^\times_\mathcal{B}(b^\times ,\pi(b^\times), b^{\times'})  {V^{\pi^\times} (b^{\times'})} . 
\end{align}
The optimal value function representing the maximum reachability probability satisfies the Bellman optimality equation, when $b^\times = \mathrm{sink}$, $V^*(b^\times) = 1$, 
when $b^\times \ne \mathrm{sink}$, 
\begin{align}\label{eq:bellman star}
V^*(b^\times) = {\textstyle\max_{a\in A} {\sum_{b^{\times'}} }  T^\times_\mathcal{B}(b^\times ,a, b^{\times'})  {V^* (b^{\times'})} }. 
\end{align}


The policy $\pi^\times(\langle b,q \rangle)$ requires us to update the belief states and the DFA state at the same time. Given current belief and DFA state $b_t,q_t$, the action is $a_t = \pi^\times(\langle b_t,q_t \rangle)$. 
\begin{lemma}\label{lem:to main}
    Given sc-iLTL objective $\varphi$ and POMDP $\mathcal{M}_\mathcal{P}$, 
    the optimal policy $\pi^\times(\langle b,q \rangle)$ on $\mathcal{M}_\mathcal{B}^\times$ maximizing reachability probability is the optimal policy $\pi^\times(\langle b, q \rangle)$ on $\mathcal{M}_\mathcal{P}$ maximizing satisfaction probability of $\varphi$. 
\end{lemma}
\begin{proof}
    Given DFA $\mathcal{A}_\varphi$ with initial DFA state $q_0=q_\mathrm{init}$, following the same policy $\pi^\times$, the satisfaction probability on POMDP equals the reachability probability on product belief MDP, 
    $ Pr_{\pi^\times}(\sigma \models \varphi)   = Pr_{\pi^\times}(\sigma^\times \models \lozenge F^\times | \sigma^\times_0 = \langle \sigma_0,q_0 \rangle)$. 
    Since each belief path $\sigma$ corresponds to path $\sigma^\times$ where $\sigma^\times_0 = \langle \sigma_0, q_0\rangle$, 
    Lemma~\ref{lem: V to R} and \eqref{eq:bellman star} finish the proof. 
\end{proof}

\subsection{Monte Carlo tree search}\label{sec:MCTS}
Here, we use MCTS to find the optimal policy on the product belief MDP. 
Similar to \cite{silver2010montecarlo}, MCTS only searches a finite number of belief states that are reachable from $b_0$. 

We use an equivalent product history state $h^\times=\langle h ,q \rangle$ to represent each product belief state $b^\times =\langle \mathcal{B}(h) ,q \rangle$ by~\eqref{eq:belief update}. 
The transition between the product history states $T^\times_\mathcal{H}(h^\times ,a, h^{\times'})$ is derived from $T^\times_\mathcal{B}(b^\times ,a, b^{\times'})$ and $\mathcal{B}(h)$. 
 
The decision tree is constructed based on the \mbox{Bellman} optimality equation \eqref{eq:bellman star}.  
Each state node~$G(h^\times) = \langle N(h^\times), \hat{V}(h^\times), \mathcal{B}(h^\times)\rangle$ contains $N(h^\times)$ counts the number of times that $h^\times$ has been visited, $\hat{V}(h^\times)$ is the value estimation of $\mathcal{B}(h^\times)$ by the mean return of all simulations starting with $\mathcal{B}(h^\times)$.
Here we use $\mathcal{B}(h^\times)$ to represent $ \mathcal{B}(h)$ where $h$ is inside $h^\times$. 
New nodes are initialized to $\langle \hat{V}_{init}(h^\times), N_{init}(h^\times) ,\mathcal{B}(h^\times)\rangle $ if the knowledge is available, and to $\langle 0, 0 ,\mathcal{B}(h^\times)\rangle$ otherwise. 
The children to state node $G(h^\times)$ are action nodes $G(h^\times a) = \langle N(h^\times, a) , \hat{V}(h^\times a), \mathcal{B}(h^\times)\rangle$. 
The children to the action node are state node $G(h^{\times’})$ where $h^{\times’}$ 
is reachable via action $a$.

MCTS adopts the upper confident bound approach to balance exploration and exploitation. 
In each simulation, actions are selected by $\arg\max_a \big\{\hat{V}(h^\times a) + c \sqrt{\frac{\log N(h^\times)}{N(h^\times,a)}} \big\}$ if we have non-zero values for all the child nodes.
Otherwise, a rollout policy $\pi_{rollout}$ (uniform random action selection) is applied. 
Each simulation will add one new node to the tree. 
When the search is complete, the algorithm returns the optimal action for $h^\times$ with high probability. 

\begin{algorithm}
\caption{\small MCTS for the Product Belief MDP}\label{alg:example}
\begin{algorithmic}
    \small
    \STATE \textbf{Procedure} Search($h^\times = \langle h,q \rangle$)
    \STATE \quad \textbf{repeat}
    \STATE \quad \quad \textbf{if} $h = \emptyset$ \textbf{then} $s \sim s_\mathrm{init}$ \textbf{else} $s \sim \mathcal{B}(h)$
    \STATE \quad \quad \textbf{end if}
    \STATE \quad \quad Simulate($s, h^\times, 0$)
    \STATE \quad \textbf{until} Timeout()
    \STATE \quad \textbf{return} $\arg\max_{a} \hat{V}(h^\times  a)$
    \STATE \textbf{end procedure}
    \vspace{0.5em}
    \STATE  \textbf{Procedure} Rollout($s, h^\times, depth$)
    \STATE \quad \textbf{if} $depth \ge d_{\text{max}}$ \textbf{or} $h^\times = {\mathrm{sink}}$ \textbf{then}
    \STATE \quad \quad \textbf{return} 0
    \STATE \quad \textbf{end if}
    \STATE \quad $a = \pi_{rollout}(\mathcal{B}(h))$, $s' \sim T_\mathcal{P}(s,a,\cdot)$
    \STATE \quad $h^{\times'} \sim T^\times_\mathcal{H}(h^\times ,a, \cdot)$
    \STATE \quad \textbf{return} $r(\mathcal{B}(h)) +\text{Rollout}(s', h, \text{depth}+1)$
    \STATE \textbf{end procedure}
    \vspace{0.5em}
    \columnbreak
    \STATE \textbf{Procedure} Simulate($s, h^{\times}, depth$)
    \STATE \quad \textbf{if} $depth \ge d_{\text{max}}$ \textbf{or} $h^\times = {\mathrm{sink}}$ \textbf{then}
    \STATE \quad \quad \textbf{return} 0
    \STATE \quad \textbf{end if}
    \STATE \quad \textbf{if} $h^\times \notin G$ \textbf{then}
    \STATE \quad \quad \textbf{for all} { $a \in A$}
    \STATE   \quad \quad \quad $G(h^\times, a)\leftarrow \langle \hat{V}_{init}(h^\times), N_{init}(h^\times) ,\mathcal{B}(h^\times)\rangle$
    \STATE \quad \quad \textbf{end for}
    \STATE \quad \quad \textbf{return} Rollout($s, h^\times, depth$)
    \STATE \quad \textbf{end if}
    \STATE \quad $a \leftarrow \arg\max_{d}\left(\hat{V}(h^\times, d) + c \cdot \sqrt{\frac{\log N(h^\times)}{N(h^\times, d)}}\right)$
    \STATE \quad $s' \sim T_\mathcal{P}(s,a,\cdot)$, $h^{\times'} \sim P^\times_\mathcal{H}(h^\times ,a, \cdot)$
    \STATE \quad $R \leftarrow r(\mathcal{B}(h)) + \text{Simulate}(s', h^{\times'}, depth+1)$ 
    \STATE \quad $N(h^\times) \leftarrow N(h) + 1$, $N(h^\times, a) \leftarrow N(h, a) + 1$
    \STATE \quad $\hat{V}(h^\times, a) \leftarrow \hat{V}(h, a) + \frac{R - \hat{V}(h^\times, a)}{N(h^\times, a)}$
    \STATE \quad \textbf{return} $R$
    \STATE \textbf{end procedure}
\end{algorithmic}
\end{algorithm}


\section{Convergence in probability and Scalability}

In this section, we justify the correctness and scalability of our algorithm.

The sc-iLTL objective induces an infinite-horizon reachability problem~\eqref{eqn: value}, which is not compatible with existing MCTS convergence guarantees. These guarantees require either a discount factor or a finite horizon.

To leverage existing MCTS convergence results, we approximate the infinite-horizon reachability problem with a finite-horizon reachability problem. 
We then justify the scalability of our approach by showing that introducing DFA states does not increase the sample complexity.
\vspace{-5pt}
\subsection{Convergence in probability for finite-horizon reachability}
Let $D$ denote the time horizon. The value function at time step $d$ for the finite-horizon reachability problem is defined recursively for $d = D-1, \dots, 0$ as, 
\begin{align}
V_D^{\pi^\times}(b^\times,d) =  \sum_{b^{\times'}} T^\times_\mathcal{B} (b^\times ,\pi(b^\times,d), b^{\times'} ) \, V_{D}^{\pi^\times}(b^{\times'},d+1), \notag
\end{align}
with terminal condition $V_D^{\pi^\times}(b^\times,D) = 1$, if $b^\times \in \mathrm{sink}$ and $V_D^{\pi^\times}(b^\times,D) = 0$ otherwise.  

Although the above \mbox{Bellman} equation is defined over the continuous belief space, we only need the value for states reachable from the initial state $\langle b_0,q_0 \rangle$ within $D$ steps to obtain the optimal policy. 
For a finite-horizon finite-state MDP with an undiscounted but bounded reward function, convergence in probability to the optimal policy using MCTS is guaranteed by Theorem 3,5 in \cite{kocsis2006bandit}. 

\subsection{Extending to infinite-horizon reachability}
Compared to a standard approximation of  $V^\pi$ by extending the horizon $D$ of $V_D^\pi$, in our case, the number of states grows when the horizon $D$ grows. 
However, such approximation still holds since $V_D^\pi$ is non-decreasing and bounded as $D$ increases. 

\begin{lemma}\label{lem: fin to inf}
The finite-horizon value function converges pointwise to the infinite-horizon value function as
\begin{align}
    \lim_{D\to \infty} V_D^\pi (\mathcal{B}(h^\times_t),t) =  V^\pi  (\mathcal{B}(h^\times_t)). 
\end{align}  
Given $\epsilon>0$, there exists a $\Bar{D} \in \mathbb{N}^+$ such that, 
$V^\pi(\mathcal{B}(h^\times_t)) - V_D^\pi (\mathcal{B}(h^\times_t),t)  < \epsilon $ for all $D\ge \Bar{D}$. 
\end{lemma}
\begin{proof}
    i) $V_D^\pi (\mathcal{B}(h^\times_t),t) \le V_{D+1}^\pi (\mathcal{B}(h^\times_t),t)$. Since any path reaches the $\mathrm{sink}$ taken into account by $V_D^\pi (\mathcal{B}(h^\times_t),t)$ is also taken into account by $V_{D+1}^\pi (\mathcal{B}(h^\times_t),t)$. 
    ii) $V_D^\pi (\mathcal{B}(h^\times_t),t)$ is bounded by $1$. 
    Thus $\lim_{D\to \infty} V_D^\pi (\mathcal{B}(h^\times_t),t)$ exists. 
    iii) By definition of $ V^\pi (\mathcal{B}(h^\times_t))$, this lemma holds. 
\end{proof}

\begin{theorem}\label{proof:MCTS}
Let $d_{max} =  \Bar{D}$ in MCTS, the probability of failure to return an optimal action given the history state $h^\times_t$ converges to zero as $\lim_{N(h^\times_t)\to\infty} Pr(\pi(\mathcal{B}(h^\times_t) \ne a^*(h^\times_t)) = 0$, 
and the estimation error of the value function is bounded by $\vert \mathbb{E}[V^*(\mathcal{B}(h^\times_t))-\hat{V} (h^\times_t)] \vert \le O \big( \frac{\log (N(h^\times_t)}{N(h^\times_t)}\big)$. 
\end{theorem}
\begin{proof}
    By Lemma~\ref{lem: fin to inf}, $V_{\Bar{D}}^\pi (\mathcal{B}(h^\times_t),t) =  V^\pi  (\mathcal{B}(h^\times_t))$. 
    Convergence in probability to the policy maximizing $V_{\Bar{D}}^\pi$ is guaranteed by Theorem 3,5 in \cite{kocsis2006bandit}. 
    Thus, convergence in probability to the policy maximizing $V^\pi$ is guaranteed. 
\end{proof}
Now we finish the proof for our main result, Theorem~\ref{thm:4}. 
Given the result of MCTS on the product belief MDP, we can recover the policy on the POMDP. 

\hspace*{1em}\emph{Proof of Theorem~\ref{thm:4}: }
    The finding of a non-Markovian policy $\pi_1^*$ on the POMDP maximizing the satisfaction probability of the sc-iLTL objective is guaranteed by 
    i) Lemma~\ref{lem:to main} shows such a $\pi_1^*$ can be constructed by a policy $\pi_2^*$ using belief state and automaton state as input, where $\pi_2^*$ 
    maximizes the reachability probability to the DFA accepting states.  
    ii) Theorem~\ref{proof:MCTS} guarantees that MCTS returns $\pi_2^*$. 
\hfill $\blacksquare$

\vspace{-3pt}
\subsection{Scalability}
\vspace{-3pt}
The scalability of our MCTS method matches that of \mbox{PO-UCT} \cite{silver2010montecarlo}, which has been demonstrated to be competitive with state-of-the-art POMDP solvers \cite{kurniawati2022partially} in handling large state spaces. This is because the sample complexity of the method is independent of the size of the belief space. Consequently, augmenting the belief space with the DFA state space does not increase the sample complexity.
The sample complexity is determined by the branching factor $K$ and depth of the tree $D$ as $\mathcal{O}((KD \log(n)+K^D)/ n)$ (Theorem 7 in \cite{kocsis2006bandit}). 
In \mbox{PO-UCT}, where each node represents a history state, the branching factor \( K \) equals the number of available actions multiplied by the number of observations. 
In our case, each action–observation pair deterministically leads to a specific DFA state, so the branching factor \( K \) remains the same. 
The depth \( D \) is a user-defined parameter; setting \( D \ge \Bar{D} \) ensures accuracy, while a smaller \( D \) trades accuracy for efficiency.

\vspace{-5pt}
\section{Case Study}\label{sec:case}
To show the expressiveness of the iLTL objective, we introduce the drone-probing problem. 
We executed 100 \mbox{independent} runs of the experiments and evaluated the performance by the percentage of successful simulations. 

\vspace{-5pt}
\subsection{Drone-probing problem and sc-iLTL objective}
\vspace{-5pt}
The drone-probing problem is in a $4\times 4$ grid world with $256$ hidden states. 
The \textbf{drone} on the grid has an action set $\{\mathrm{N},\mathrm{S},\mathrm{E},\mathrm{W},\mathrm{X}\}$ that stands for moving to the four directions or staying. 
The drone is initialized on the grid $(0,0)$. 
The \textbf{ground target} moves randomly on the grid as we assign equal probability to all the possible following locations of the target. 
The target is initialized on the grid, excluding the corner, with equal probability. 
The \textbf{sensor} is limited as described in section~\ref{sec:example}. 

The sc-iLTL objective allows us to specify the following task. 
Suppose we want the drone to use its imperfect sensor to accurately locate the ground target and then move to the landing zone $(3,3)$. 
Reaching the landing zone before getting an accurate measure is seen as a failure. 
Not reaching the landing zone within the simulation horizon $100$ is seen as a failure. 
We define the control objective as, 
\begin{align}\label{eq:obj}
   \varphi =  &\lozenge\bigg(\bigvee_{i=1}^{256} \mathrm{ineq}_{\mathrm{measure}, i}\bigg) \wedge \lozenge(\mathrm{ineq}_\mathrm{goal}) \notag \\ 
   &\wedge (\neg \mathrm{ineq}_\mathrm{goal} \mathcal{U} \bigvee_{i=1}^{256} \mathrm{ineq}_{\mathrm{measure}, i}), 
\end{align}
where $\Lambda = \{\mathrm{ineq}_{\mathrm{measure}, 1}, \mathrm{ineq}_{\mathrm{measure}, 2}, ...,  \mathrm{ineq}_{\mathrm{measure}, 256},\allowbreak \mathrm{ineq}_\mathrm{goal}\}$ is the set of atomic propositions, 
each $\mathrm{ineq}_{\mathrm{measure}, i}:= p_i ^T b>0.9$ measures the confidence of the location of the ground target, $p_i$ is a vector with all zero elements except the $i$-th entry to be $1$,   
$\mathrm{ineq}_\mathrm{goal} := p_g ^T b \ge 1$ stands for the event of reaching the landing zone. 

Specifying this task with a standard LTL formula defined over the hidden state space is hard. 
Even if the drone is right above the target,  $\Vert b\Vert_\infty$ may still be below~$0.9$. 
 
\subsection{Experiment}
The experiments were conducted using Python on a Windows 11 machine equipped with an Intel i9-14900K processor.
The Mona package \cite{fuggitti-ltlf2dfa} was used to translate the iLTL objective into a DFA. 
The DFA has $4$ states, thus increasing the number of all reachable belief states by $4$ times. 
We executed 100 independent runs. 
In each run, the drone and ground target are spawned based on the initial belief. 
For each MCTS search, we apply $2000$ simulations with tree depth $d_{max}=20$.

Figure~\ref{fig:exp} illustrates algorithm performance using a histogram and the average evolution of $\Vert b \Vert_\infty$.
Out of the 100 runs, 87 were successful.
On average, successful runs required $40.71$ steps to complete the task.
$4$ out of $13$ failed experiments are due to exceeding the horizon. 
Other failed experiments were due to the randomness of MCTS output. 
\begin{figure}[ht]
    \centering
    \includegraphics[width=0.3\textwidth]
    {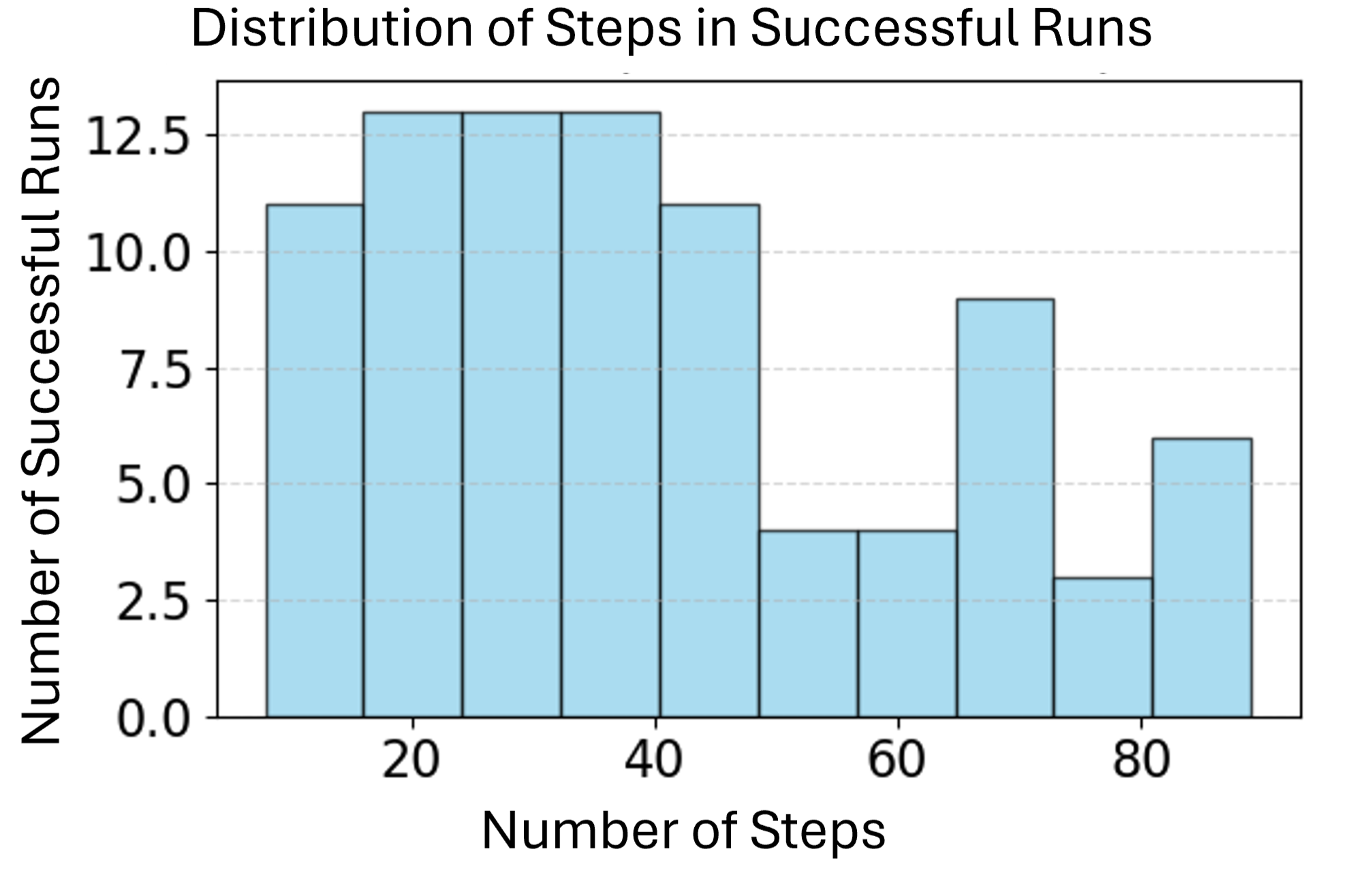}
    \hfill
    \includegraphics[width=0.3\textwidth]{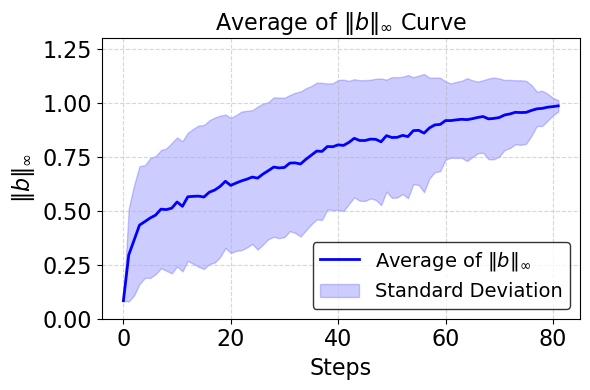}
    \caption{Histograms of steps in successful runs (left) and change of average $\Vert b \Vert_\infty$ during successful runs (right). We fix $\Vert b \Vert_\infty$ once it goes beyond $0.9$. 
    Each action is returned by MCTS using $2000$ simulations with depth $d_{max}=20$. 
    Our MCTS algorithm achieves an $87\%$ success rate for $\varphi$. 
    We see the MCTS algorithm tries to increase $\Vert b \Vert_\infty$ up to the confidence predefined and finishes the task in an average of $40.71$ steps.  
    }
    \label{fig:exp} 
\end{figure}

\section{Conclusion}
This work proposes an sc-iLTL objective for POMDPs, for which the optimal policy can be synthesized via MCTS. 
The iLTL objective defined in belief space is more expressive than the LTL objective defined using hidden states.
Specifically, we utilize inequality in belief space to specify objectives on POMDPs. 
Sc-iLTL is suitable for objectives related to safety and surveillance. 
The transformation of the control for the sc-iLTL objective into a reachability problem on the product belief MDP enables us to leverage MCTS to find optimal policies effectively. 
Experiments in the drone\mbox{-}probing problem demonstrate the expressiveness of the sc\mbox{-}iLTL objective and the performance of our MCTS method. 

\vspace{5pt}
\bibliographystyle{IEEEtran_local}
\bibliography{references_cleaned_for_ieee,IEEEabrv}
\end{document}